%% file: article.tex
\documentclass[runningheads]{llncs}

\input{article-tex/prelude.tex}
\input{article-tex/examples.tex}

\begin{document}
\title{Adjoint Reactive GUI Programming}

 \author{Christian Uldal Graulund\inst{1}\orcidID{0000-0003-3297-9471} \and
   Dmitrij Szamozvancev\inst{2} \and
   Neel Krishnaswami\inst{2}\orcidID{0000-0003-2838-5865}
}
\authorrunning{Graulund et al.}
 \institute{IT University of Copenhagen, 2300 Copenhagen, DK \email{cgra@itu.dk} \and
   University of Cambridge, Cambridge CB3 0FD, UK \email{nk480@cl.cam.ac.uk,ds709@cl.cam.ac.uk}}

\maketitle              % typeset the header of the contribution
\begin{abstract}
  Most interaction with a computer is done via a graphical user
  interface.  Traditionally, these are implemented in an imperative
  fashion using shared mutable state and callbacks. This is efficient,
  but is also difficult to reason about and error prone. Functional
  Reactive Programming (FRP) provides an elegant alternative which
  allows GUIs to be designed in a declarative fashion. However, most
  FRP languages are synchronous and continually check for new
  data. This means that an FRP-style GUI will ``wake up'' on each
  program cycle. This is problematic for applications like text
  editors and browsers, where often nothing happens for extended
  periods of time, and we want the implementation to sleep until new
  data arrives. In this paper, we present an \emph{asynchronous} FRP
  language for designing GUIs called $\langname$. Our language
  provides a novel semantics for widgets, the building block of GUIs,
  which offers both a natural Curry--Howard logical interpretation and
  an efficient implementation strategy.

\end{abstract}
\section{Introduction}
\label{sec:intro}
\input{article-tex/intro.tex}

\section{The Language}
\label{sec:language}
\input{article-tex/language.tex}

\section{Formal Calculus}
\label{sec:formal_calculus}
\input{article-tex/formal_calculus.tex}

\section{Formal Semantics}
\label{sec:formal_semantics}
\input{article-tex/formal_semantics.tex}

\section{Related and Future Work}
\label{sec:related_work}
\input{article-tex/related_work.tex}

\bibliographystyle{splncs04}
\bibliography{article-tex/bibliography}
\end{document}

%% file: article-tex/prelude.tex
\usepackage[scale=0.7]{geometry}
\usepackage{
  amsmath,amssymb,amsfonts,
  stmaryrd,mathpartir,
  xspace,xcolor,wasysym,
  tikz-cd,mathtools,hyperref,ifthen}

\newcommand{\hastypei}[3]{#1 \vdash_{i} #2 : #3}
\newcommand{\hastypec}[4]{#1;#2 \vdash_{c} #3 : #4}
\newcommand{\hastypel}[5]{#1;#2;#3 \vdash_{l} #4 : #5}
\newcommand{\hastypelt}[6]{#1;#2;#3 \vdash_{l} #4 :_{#5} #6}

\newcommand{\dom}[1]{\sym{dom}(#1)}
\newcommand{\attime}[2]{#1 \downarrow^{#2}}

\newcommand{\wfcxti}[1]{\vdash_i #1}
\newcommand{\wfcxtc}[2]{#1 \vdash_c #2}
\newcommand{\wfcxtl}[2]{#1 \vdash_l #2}

\newcommand{\Time}{\sym{Time}}
\newcommand{\Id}{\sym{Id}}
\newcommand{\CUnit}{1}
\newcommand{\LUnit}{\sym{I}}
\newcommand{\Event}{\Diamond}
\newcommand{\G}[1]{\sym{G}\;#1}
\newcommand{\F}[1]{\sym{F}\;#1}
\newcommand{\Widget}{\sym{Widget}}
\newcommand{\Str}[1]{\sym{Str}\;#1}

\newcommand{\lunit}{\langle \rangle}
\newcommand{\cunit}{\star}
\newcommand{\letterm}[3]{\sym{let}\;#1 = #2\;\sym{in}\;#3}
\newcommand{\letunit}[2]{\letterm{\lunit}{#1}{#2}}
\newcommand{\letunitt}[3]{\letterm{\lunit \at #1}{#2}{#3}}
\newcommand{\letpair}[4]{\letterm{(#1,#2)}{#3}{#4}}
\newcommand{\letpairt}[5]{\letterm{(#1,#2) \at #3}{#4}{#5}}
\newcommand{\pair}[2]{\langle #1,#2 \rangle}
\newcommand{\event}{\sym{evt}}
\newcommand{\letevent}[3]{\letterm {\event\;#1}{#2}{#3}}
\newcommand{\at}{\mathop{@}}
\newcommand{\letat}[4]{\letterm{#1 \at #2}{#3}{#4}}
\newcommand{\runG}{\sym{runG}}
\newcommand{\letF}[3]{\letterm{\F{#1}}{#2}{#3}}
\newcommand{\exterm}[2]{\{#1,#2\}}
\newcommand{\letexists}[4]{\letterm{\sym{unpack}\;\exterm{#1}{#2}}{#3}{#4}}
\newcommand{\selectterm}[6]{\sym{select}\;(#1\;\sym{as}\;#2 \mapsto #3 \mid #4\;\sym{as}\;#5 \mapsto #6)}

\newcommand{\sub}[3]{\{#1/#2\}#3}

\newcommand{\Cmd}{\sym{Cmd}}
\newcommand{\comp}{\bowtie}
\newcommand{\Set}{\sym{Set}}
\newcommand{\rcat}{\mathcal{R}}
\newcommand{\relcat}{\sym{Rel}_\rcat}
\newcommand{\Later}{\rhd}
\newcommand{\Prev}{\lhd}
\newcommand{\set}[1]{\left \{#1 \right \}}
\newcommand{\setcom}[2]{\set{#1 \mid #2}}
\newcommand{\nats}{\mathbb{N}}

\newcommand{\sym}[1]{\ensuremath{\mathsf{#1}}}
\newcommand{\langname}{\lambda_{\Widget}}

%% file: article-tex/intro.tex
Many programs, like compilers, can be thought of as functions -- they
take a single input (a source file) and then produce an output (such
as a type error message). Other programs, like embedded controllers,
video games, and integrated development environments (IDEs), engage in
a dialogue with their environment: they receive an input, produce an
output, and then wait for a new input that depends on the prior input,
and produce a new output which is in turn potentially based on the
whole history of prior inputs.

Our usual techniques for programming interactive applications are
often very confusing, because the different parts of the program are
not written to interact via structured control flow (i.e., by passing
and return values from functions, or iterating over data in
loops). Instead, they communicate indirectly, by registering
state-manipulating callbacks with one another, which are then
implicitly invoked by an event loop. This makes program reasoning very
challenging, since each of these features -- aliased mutable state,
higher-order functions, and concurrency -- represents a serious
obstacle on its own, and interactive programs rely upon their
\emph{combination}.

This challenge has led to a great deal of work on better abstractions
for programming reactive systems. Two of the main lines of work on
this problem are \emph{synchronous dataflow} and \emph{functional
  reactive programming}. The synchronous dataflow languages, like
Esterel~\cite{esterel}, Lustre~\cite{lustre}, and Lucid
Synchrone~\cite{synchrone}, feature a programming model inspired by
Kahn networks. Programs are networks of stream-processing nodes which
communicate with each other, each node consuming and producing a fixed
number of primitive values at each clock tick. The first-order nature
of these languages makes them strongly analysable, which lets them
offer powerful guarantees on space and time usage. This means
they see substantial use in embedded and safety-critical contexts.

Functional reactive programming, introduced by Elliott and
Hudak~\cite{FRAN}, also uses time-indexed values, dubbed signals,
rather than mutable state as its basic primitive. However, FRP differs
from synchronous dataflow by sacrificing static analysability in
favour of a much richer programming model. Signals are true
first-class values, and can be used freely, including in higher-order
functions and signal-valued signals. This permits writing programs
with a dynamically-varying dataflow network, which simplifies writing
programs (such as GUIs) in which the available signals can change as
the program executes. Over the past decade, a long line of work has
refined FRP via the Curry--Howard
correspondence~\cite{krishnaswami2011ultrametric,Jeltsch2012,jeffrey2012,jeltsch2013temporal,krishnaswami13frp,cave14fair,bahr2019simply}. This
approach views functional reactive programs as the programming
counterpart for proofs of formulas in linear temporal
logic~\cite{LTL}, and has enabled the design of calculi which can rule
out spacetime leaks~\cite{krishnaswami13frp} or can enforce temporal
safety and liveness properties~\cite{cave14fair}.

However, both synchronous dataflow and FRP (in both original and modal
flavours) have a \emph{synchronous} (or ``pull'') model of time --
time passes in ticks, and the program wakes up on every tick to do a
little bit more computation. This is suitable for applications in
which something new happens at every time step (e.g., video games),
but many GUI programs like text editors and spreadsheets spend most of
their time doing nothing. That is, even at each event, most of the
program will continue doing nothing, and we only want to wake up a
component when an event directly relevant to it occurs. This is
important both from a performance point of view, as well as for saving
energy (and extending battery life). Because of this need, most GUI
programs continue to be written in the traditional
callbacks-on-mutable-state style.

In this paper, we give a reactive programming language whose type
system both has a very straightforward logical reading, and which can
give natural types to stateful widgets and the event-based programming
model they encourage. We also derive a denotational semantics of the
language, by first working out a semantics of widgets in terms of the
operations that can be performed upon them and the behaviour they
should exhibit. Then, we find the categorical setting in which the
widget semantics should live, and by studying the structure this setup
has, we are able to interpret all of the other types of the
programming language.

\subsection{Contributions}

The contributions of this paper are as follows:
\begin{itemize}
\item We give a descriptive semantics for widgets in GUI programming,
  and show that this semantics correctly models a variety of expected
  behaviours. For example, our semantics shows that a widget which is
  periodically re-set to the colour red is different from a widget
  that was only persistently set to the colour red at the first
  timestep. Our semantic model can show that as long as neither one is
  updated, they look the same, but that they differ if they are ever
  set to blue -- the first will return to red at reset time, and the
  second will remain blue.

\item From this semantics, we find a categorical model within which
  the widget semantics naturally fits. This model is a Kripke--Joyal
  presheaf semantics, which is morally a ``proof-relevant'' Kripke
  model of temporal logic.

\item We give a concrete calculus for event-based reactive
  programming, which can be implemented in terms of the standard
  primitives for modern GUI programming, scene graphs (or DOM) which
  are updated via callbacks invoked upon events. We then show that our
  model can soundly interpret the types of our calculus in an entirely
  standard way, showing that the types of our reactive programming
  language can be interpreted as time-varying sets.

\item Furthermore, this calculus has an entirely standard logical
  reading in terms of the Curry--Howard correspondence. It is a
  ``linear temporal linear logic'', with the linear part of the
  language corresponding to the
  Benton--Wadler~\cite{Benton1996LinearLM} LNL calculus for linear
  logic, and the temporal part of the language corresponding to S4.3
  linear temporal logic.  We also give a proof term for the S$_t$4.3
  axiom enforcing the linearity of time, and show that it corresponds
  to the \texttt{select} primitive of concurrent programming.
\end{itemize}

%%% Local Variables:
%%% mode: latex
%%% TeX-master: "fossacs-paper-2021.tex"
%%% End:

%% file: article-tex/language.tex
In this section we give an informal presentation of $\langname$
through the API of the $\sym{Widget}$ type. This API mirrors how one
would work with a GUI at the browser level. An important feature of a
well-designed GUI is that it should not do anything when not in
use. In particular, it should not check for new inputs in each program
cycle (\emph{pull}-based reactive programming), but rather sleep until
new data arrives (\emph{push}-based reactive programming). Many FRP
languages are \emph{synchronous} languages and have some internal
notion of a timestep. These languages are mostly pull-based, whereas
more traditional imperative reactive languages are push-based. The
former have clear semantics and are easy to reason about, the latter
have efficient implementations. In $\langname$ we would like to
combine these aspects and get a language that is easy to reason about
with an efficient implementation.

In general, we think of a widget as a \emph{state through time}, i.e.,
at each timestep, the widget is in some state which is presented to
the user. The widget is modified by \emph{commands}, which can update
the state. To program with widgets, the programmer applies commands at
various times.

The proper type system for a language of widgets should thus be a
system with both state and time. If we consider what a \emph{logic}
for widgets should be, there are two obvious choices. A logic for
state is linear logic~\cite{girad1987lin}, and a logic for time is
linear temporal logic~\cite{LTL}. The combination of these
two is the correct setting for a language of widgets, and, going
through Curry--Howard, the corresponding type theory is a linear,
linear temporal type theory.

\subsection{Widget API}
To work with widgets, we define a API which mirrors how one would work
with a browser level GUI:

\widgetAPI
\noindent
The first two commands creates and deletes widgets, respectively. The
$\multimap$ should be understood as \emph{state passing}. We read the
type of $\sym{newWidget}$ as ``consuming no state, produce a new
identifier index and a widget with that identifier index''. The
identifier indices are used to ensure the correct behavior when using
the \sym{split} and \sym{join} commands explained below. The existential
quantification describes the \emph{non-deterministic} creation of an
identifier index. The use of non-determinism is crucial in our
language and will be explaining in further detail in
\autoref{sec:formal_semantics}. Since $\langname$ has a linear type
system, we need an explicit construction to delete state. For widgets,
this is $\sym{dropWidget}$. The type is read as ``for any identifier
index, consume a widget with that identifier index and produce
nothing''.

The first command that modifies the state of a widget is
$\sym{setColor}$. Here we see the adjoint nature of the calculus with
$\sym{F}\;\sym{Color}$. A color is itself \emph{not} a linear thing,
and as such, to use it in the linear setting, we apply $\sym{F}$,
which moves from the non-linear (Cartesian) fragment and into the
linear fragment. The second new thing is the linear product
$\otimes$. This differs from the regular non-linear product in that we
do not have projection maps. Again, because of the linearity of our
language,  we cannot just discard state. We can now read the type of
$\sym{setColor}$ as ``Given a color and a identified widget, consume
both and produce a new widget''. The produced widget is the same as
the consumed widget, but with the color attribute updated.

The next two commands, $\sym{onClick}$ and $\sym{onKeypress}$, are
roughly similar. Both register a handle on the widget, for a mouse
click and a key press, respectively. Here we see the first use of the
$\Diamond$ modality, which represents an \emph{event}. The type
$\Diamond A$ represents that \emph{at some point in the future} we
will receive something of type $A$. Importantly, because of the
asynchronous nature of $\langname$, we do not know \emph{when} it
happens. We can then read the type of $\sym{onClick}$ as ``Consuming
an identified widget, produce an updated widget together with a mouse
click event''. The same holds for the type of $\sym{onKeypress}$
except a key press event is produced.

The two commands $\sym{out}$ and $\sym{into}$ allows us to work with
events in a more precise way. Given an event, we can use $\sym{out}$
to ``unfold'' it into an existential. The $\at$ connective describes a
type that is only available at a certain timestep, i.e., $A \at n$
means ``at the timestep $n$, a term of type $A$ will be
available''. The $\sym{into}$ commands is the reverse of $\sym{out}$
and turns an existential and an $\at$ into an event.

So far, we have only applied commands to a widget in the current
timestep, but to program appropriately with widgets, we should be able
to react to events and apply commands ``in the future''. This is
exactly what the $\sym{split}$ and $\sym{join}$ commands allows us to
do. The type of $\sym{split}$ is read as ``Given any time step and any
identified widget, split the widget into all the states \emph{before}
that time and the widget \emph{at} that time''. We denote the
collection of states before a given time a \emph{prefix} and give it
the type $\sym{Prefix}$. Given the state of the widget at a given
timestep, we can now apply commands \emph{at that timestep}. Note that
both the prefix and the widget is indexed by the same identifier
index. This is to ensure that when we use $\sym{join}$, we combined
the correct prefix and future.

\subsection{Widget Programming}
To see the API in action, we now proceed with several examples of
widget programming. For each example, we will add a comment on each
line with the type of variables, and then explain the example in
text afterwards.

One of the simplest things we can do with a widget is to perform some
action when the widget is clicked. In the following example, we
register a handler for mouse clicks, and then we use the click event
to change the color of the widget to red at the time of the click. To
do this, we use the $\sym{out}$ map to get the time of the event, then
we split the widget and apply $\sym{setColor}$ at that point in the
future.

\turnRedOnClick
\noindent
To see why this type checks, we go through the example line by line.
In line 3, we register a handle for a mouse click on the widget. In
line 4, we turn the click event into an existential. In line 5, we get
$c_2$ which is a binding that is only available at the timestep
$x$. Since we only need the \emph{time} of the click, we discharge the
click itself in line 6. In line 7 and 8, we split the
widget using the timestep $x$ and bind $w_3$ to the state of the
widget at that timestep. In line 9, we change the color of the widget
to red at $x$ and in line 10 we recompose the widget.

In general, we will allow pattern matching in eliminations and since widget
identity indices can always be inferred, we will omit them. In this
style, the above example become:

\turnRedOnClickSugar
\noindent
We will use the same sugared style throughout the rest of the
examples.

The above example turns a widget red exactly at the time of the
mouse click, but will not do anything with successive clicks. To also
handle further mouse clicks, we must register an event handler
\emph{recursively}. This is a simple modification of the previous code:

\keepTurningRed
\noindent
By calling itself recursively, this function will make sure a widget
will always turn red on a mouse click.

To understand the difference between two above examples, consider the
following code

\clickThenKeep
\noindent
where $w$ is some widget and $turnBlueOnClick$ is the obvious
modification of the above code. On the first click, the widget will
turn blue, on the second click it will turn red and on any subsequent
click, it will keep turning red, i.e., stay red unless further
modified.

When working with widgets, we will often register multiple handlers on
a single widget.  For example, a widget should have one behavior for a
click and another behavior for a key press.  To choose between two
events, we use the $\sym{select}$ construction.  This construction is
central to our language and how to think about a push-based reactive
language.

Given two events, $t_1 : \Diamond A, t_2 : \Diamond B$, there are
three possible behaviors: Either $t_1$ returns first, and we wait for
$t_2$ or $t_2$ returns first and we wait for $t_1$ or they return at
the same time. In general, we want to select between $n$ events, but
if we need to handle all possible cases, this will give $2^n$ cases,
so to keep the syntax linear in size, we will omit the last case. In
the case events \emph{actually} return at the same time, we do a
non-deterministic choice between them. The syntax for $\sym{select}$
is
\begin{align*}
  \selectterm{t_1}{x}{t_1'}{t_2}{y}{t_2'}
\end{align*}
where
$x : A, y : B, t_1' : A \multimap \Diamond B \multimap \Diamond C$ and
$t_2' : B \multimap \Diamond A \multimap \Diamond C$. The second
important thing to understand when working with $\sym{select}$ is that
given we are working with events, we do not actually know at which
timestep the events will trigger, and hence, we do not know what the
(linear) context contains. Thus, when using $\sym{select}$, we will
\emph{only} know either $a: A, t_2:\Diamond B$ or
$t_1 : \Diamond A, b : B$. We can think of the $\sym{select}$ rule a
\emph{case-expression} that must respect time.

In the following example, we register two handlers, one for clicks and
one for key presses, and change the color of the widget based on which
returns first.

\widgetSelect
\noindent
In line 3 and 4, we register the two handlers. In line 5-9, we use the
\sym{select} construction. In the first case, the click happens first
and we return the color red. In the second case, the key press happens
first and we return the color blue. In both cases, because of the
linear nature of the language, we need to do a let-binding to
discharge the unit and the char, respectively. In line 10, we turn the
color event into an existential. In line 11, we use the timestep of the
color event to split the widget, and in line 12, we change the color
of the widget at that time and recompose it.

To see how $\langname$ differs from more traditional synchronous FRP
languages, we will examine how to encode a kind of streams. Since
our language is \emph{asynchronous}, the stream type must be encoded as
\begin{equation*}
  \Str{A} := \nu \alpha. \Diamond (A \otimes \alpha)
\end{equation*}
This asynchronous stream will \emph{at some point in the future} give
a head and a tail. We do not know when the first element of the stream
will arrive, and after each element of the stream is produced, we will
wait an indeterminate amount of time for the next element. The reason
why the stream type in $\langname$ must be like this is essentially
that we want a \emph{push-based} language, i.e., we do not want to
wake up and check for new data in each program cycle. Instead, the
program should sleep until new data arrives.

To show the difference between the asynchronous stream and the more
traditional synchronous stream, we will look at some examples. With a
traditional stream, a standard operation is zipping two streams: that
is, given $\Str{A}$ and $\Str{B}$, we can produce $\Str{A \times B}$,
which should be the element-wise pairing of the two streams. It should
be clear that this is not possible for our asynchronous streams.
Given two streams, we can wait until the first stream produces an
element, but the second stream may only produce an element after a
long period of time. Hence, we would need to buffer the first element,
which is not supported in general. Remember, when using
$\sym{select}$, we can not use any already defined linear variables,
since we do not know if they will be available in the future.

Rather than zipping stream, we can instead do a kind of
\emph{interleaving} as shown below. We use $\sym{fold}$ and
$\sym{unfold}$ to denote the folding and unfolding of the fixpoint.

\interleaveEx
\noindent
Here, we use $\sym{select}$ to choose between which stream returns
first, and then we let that element be the first element of the new
stream.

On the other hand, some of the traditional FRP functions on streams
can be translated. For instance, we can map of function over a stream,
given that \emph{it is available at each step in time}:
\streamMap
\noindent
The type $\sym{F}(\sym{G} (A \multimap B))$ is read as a linear
function with no free variables that can be used in a non-linear
fashion, i.e., duplicated.  This restriction to such ``globally
available functions'' is reminiscent of the ``box'' modality in Bahr
et al.~\cite{bahr2019simply} and Krishnaswami~
\cite{krishnaswami13frp}, and the $\sym{F}$ and $\sym{G}$ construction
can be understood as decomposing the box modality into two separate
steps. This relationship will be made precise in the logical
interpretation of $\langname$ in \autoref{sec:formal_calculus}

As a final example, we will show how to dynamically update the GUI,
i.e., how to add new widgets on the fly. Before we can give the
example, we need to extend our widget API, to allow composition of
widgets. To that end, we add the $\sym{vAttach}$ command to our
API.
\begin{equation*}
  \sym{vAttach} : \forall(i,j:\sym{Id}), \sym{Widget}\;i \multimap \sym{Widget}\;j \multimap \sym{Widget}\;i
\end{equation*}
This command should be understood as an abstract version the
$\sym{div}$ tag in HTML. In the following example, we think of the
widget as a simple button that when clicked, will create a new
button. When \emph{any} of the buttons gets clicked, a new button gets
attached.

\buttonStack
\noindent
The important step here is in line 6 and 7. Here the new button is
attached at the time of the mouse click, and $buttonStack$ is called
recursively on the newly created button.

%%% Local Variables:
%%% mode: latex
%%% TeX-master: t
%%% End:

%% file: article-tex/formal_calculus.tex
This sections gives the formal rules, the meta-theory and the logical
interpretation of $\langname$. Briefly, the language is an mixed
linear-non-linear adjoint calculus in the style of
Benton--Wadler~\cite{benton95mixed,Benton1996LinearLM}. The non-linear
fragment, also called Cartesian in the following, is a minimal simply
typed lambda calculus whereas the linear fragment contains several
non-standard judgments used for widget programming.

\subsection{Contexts and Typing Judgments}
We have three separate typing judgments: one for indices, one for
Cartesian (non-linear) terms, and one for linear terms. These are
distinguished by a subscript on the turnstile, $i$ for indices, $c$
for Cartesian terms and $l$ for linear terms. These depend on
different contexts. The index judgment depends only on a index
context, whereas the Cartesian and linear judgments depends on both an
index and a linear and/or a Cartesian context. The rules for context
formation is given in \autoref{fig:context_rules}. These are mostly
standard except for the dependence on a previously defined context and
the fact that the linear context contains variables of the form
$a :_\tau A$, i.e., temporal variables. The judgment $a :_\tau A$ is
read as ``$a$ has the type $A$ at the timestep $\tau$''. In the linear
setting we will write $a : A$ instead of $a :_0 A$, i.e., a judgment
in the current timestep.

\begin{figure}
  \centering
  \begin{mathpar}
    \begin{array}{lcc}
      \mbox{Indices:}
      & \inferrule*{~}{\wfcxti{\cdot}}
      & \inferrule*
        {\wfcxti{\Theta} \\ s \not\in \dom{\Theta} \\ \sigma \in
      \{\Id,\Time\}}
      {\wfcxti{\Theta,s : \sigma}} \\
      \mbox{Cartesian:}
      & \inferrule*{~}{\wfcxtc{\cdot}}
      & \inferrule*
    {\wfcxtc{\Theta}{\Gamma} \\ x \not\in\dom{\Gamma} \\
      \Theta \vdash_c X}
      {\wfcxtc{\Theta}{\Gamma,x:X}} \\
      \mbox{Linear:}
      & \inferrule*{~}{\wfcxtl{\cdot}}
      & \hspace{5pt} \inferrule*
    {\wfcxtl{\Theta}{\Delta} \\ x \not\in\dom{\Delta} \\
      \Theta \vdash_l A \\ \hastypei{\Theta}{\tau}{\Time}}
      {\wfcxtl{\Theta}{\Delta,a:_\tau A}}
    \end{array}
  \end{mathpar}
  \caption{Context Formation}
  \label{fig:context_rules}
\end{figure}

The index judgment describes how to introduce indices. The typing
rules are given in \autoref{fig:index_typing_rules}. The judgment
$\hastypei{\Theta}{\tau}{\sigma}$ contains a single context, $\Theta$,
for index variables. There are only two sorts of indices, identifiers
and timesteps.

\begin{figure}
  \centering
  \textbf{Index Judgments:}
    \begin{mathpar}
    \inferrule*[right=Time]
    {\tau \in \Time}
    {\hastypei{\Theta}{\tau}{\Time}}
    \and
    \inferrule*[right=Id]
    {\iota \in \Id}
    {\hastypei{\Theta}{\iota}{\Id}}
    \and
    \inferrule*[right=Var]
    {i : \sigma \in \Theta}
    {\hastypei{\Theta}{i}{\sigma}}
  \end{mathpar}
  \caption{Index Typing rules}
  \label{fig:index_typing_rules}
\end{figure}

The Cartesian judgment describes the Cartesian, or non-linear,
fragment. The typing rules are given in
\autoref{fig:cartesian_typing_rules}. This is a minimal simply typed
lambda calculus with the addition of the $\sym{G}$ type, used for
moving between the linear and Cartesian fragment, and explained
further below. The judgment $\hastypec{\Theta}{\Gamma}{t}{A}$ has two
contexts; $\Theta$ for indices and $\Gamma$ for Cartesian variables.

\begin{figure}
  \centering
  \textbf{Cartesian Judgments:}
  \begin{mathpar}
    \inferrule*[right=($\CUnit$-I)]
    {~}
    {\hastypec{\Theta}{\Gamma}{\cunit}{\CUnit}}
    \and
    \inferrule*[right=(Var)]
    {(x:X) \in \Gamma}
    {\hastypec{\Theta}{\Gamma}{x}{X}}
    \and
    \inferrule*[right=($\to$-I)]
    {\hastypec{\Theta}{\Gamma,x:X}{e}{Y}}
    {\hastypec{\Theta}{\Gamma}{\lambda x.e}{X \to Y}}
    \and
    \inferrule*[right=($\to$-E)]
    {\hastypec{\Theta}{\Gamma}{e_1}{X \to Y} \\
      \hastypec{\Theta}{\Gamma}{e_2}{X}}
    {\hastypec{\Theta}{\Gamma}{e_1e_2}{Y}}
    \and
    \inferrule*[right=($\sym G$-I)]
    {\hastypel{\Theta}{\Gamma}{\cdot}{t}{A}}
    {\hastypec{\Theta}{\Gamma}{\G{t}}{\G{A}}}
  \end{mathpar}
  \caption{Cartesian Typing rules}
  \label{fig:cartesian_typing_rules}
\end{figure}

The linear fragment is most of the language, and the typing rules are
given in \autoref{fig:linear_typing_rules}. The judgment is done w.r.t
three contexts, $\Theta$ for index variables, $\Gamma$ for Cartesian
variables and $\Delta$ for linear variables.  Many of the rules are
standard for a linear calculus, except for the presence of the
additional contexts. We will not describe the standard rules any
further.

\begin{figure}
  \centering
  \textbf{Linear Judgments:}
  \begin{mathpar}
    \inferrule*[right=(Var)]
    {~}
    {\hastypel{\Theta}{\Gamma}{a:A}{a}{A}}
    \and
    \inferrule*[right=($\multimap$-I)]
    {\hastypel{\Theta}{\Gamma}{\Delta,a:A}{t}{B}}
    {\hastypel{\Theta}{\Gamma}{\Delta}{\lambda a.t}{A \multimap B}}
    \and
    \inferrule*[right=($\multimap$-E)]
    {\hastypel{\Theta}{\Gamma}{\Delta_1}{t_1}{A \multimap B} \\
      \hastypel{\Theta}{\Gamma}{\Delta_2}{t_2}{A}}
    {\hastypel{\Theta}{\Gamma}{\Delta_1,\Delta_2}{t_1t_2}{B}}
    \and
    \inferrule*[right=($\LUnit$-I)]
    {~}
    {\hastypel{\Theta}{\Gamma}{\cdot}{\lunit}{\LUnit}}
    \and
    \inferrule*[right=($\LUnit$-E)]
    {\hastypel{\Theta}{\Gamma}{\Delta_1}{t_1}{\LUnit} \\
      \hastypel{\Theta}{\Gamma}{\Delta_2}{t_2}{C}}
    {\hastypel{\Theta}{\Gamma}{\Delta_1,\Delta_2}{\letunit{t_1}{t_2}}{C}}
    \and
    \inferrule*[right=($\otimes$-I)]
    {\hastypel{\Theta}{\Gamma}{\Delta_1}{t_1}{A} \\
      \hastypel{\Theta}{\Gamma}{\Delta_2}{t_2}{B}}
    {\hastypel{\Theta}{\Gamma}{\Delta_1,\Delta_2}{\pair{t_1}{t_2}}{A \otimes B}}
    \and
    \inferrule*[right=($\otimes$-E)]
    {\hastypel{\Theta}{\Gamma}{\Delta_1}{t_1}{A \otimes B} \\
      \hastypel{\Theta}{\Gamma}{\Delta_2, a:A, b:B}{t_2}{C}}
    {\hastypel{\Theta}{\Gamma}{\Delta_1,\Delta_2}{\letpair{a}{b}{t_1}{t_2}}{C}}
    \and
    \inferrule*[right=($\Event$-I)]
    {\hastypel{\Theta}{\Gamma}{\Delta}{t}{A}}
    {\hastypel{\Theta}{\Gamma}{\Delta}{\event\;t}{\Event A}}
    \and
    \inferrule*[right=($\Event$-E)]
    {\hastypel{\Theta}{\Gamma}{\Delta}{t_1}{\Event A}\\
      \hastypel{\Theta}{\Gamma}{a:A}{t_2}{\Event B}}
    {\hastypel{\Theta}{ }{\Delta}{\letevent{a}{t_1}{t_2}}{\Event B}}
    \and
    \inferrule*[right=(@-I)]
    {\hastypei{\Theta}{\tau}{\Time} \\
      \hastypelt{\Theta}{\Gamma}{\Delta}{t}{\tau}{A}}
    {\hastypel{\Theta}{\Gamma}{\Delta}{t \at \tau}{A \at \tau}}
    \and
    \inferrule*[right=(@-E)]
    {\hastypei{\Theta}{t}{\Time} \\
      \hastypel{\Theta}{\Gamma}{\Delta_1}{t_1}{A \at \tau} \\
      \hastypel{\Theta}{\Gamma}{\Delta_2,a:_{\tau} A}{t_2}{B}}
    {\hastypel{\Theta}{\Gamma}{\Delta_1,\Delta_2}{\letat{a}{\tau}{t_1}{t_2}}{B}}
    \and
    \inferrule*[right=($\sym{G}$-E)]
    {\hastypec{\Theta}{\Gamma}{e}{\G{A}}}
    {\hastypel{\Theta}{\Gamma}{\cdot}{\runG\;e}{A}}
    \and
    \inferrule*[right=($\sym F$-I)]
    {\hastypec{\Theta}{\Gamma}{e}{X}}
    {\hastypel{\Theta}{\Gamma}{\cdot}{\F{e}}{\F{x}}}
    \and
    \inferrule*[right=($\sym F$-E)]
    {\hastypel{\Theta}{\Gamma}{\Delta_1}{t_1}{\F{X}} \\
      \hastypel{\Theta}{\Gamma,x:X}{\Delta_2}{t_2}{B}}
    {\hastypel{\Theta}{\Gamma}{\Delta_1,\Delta_2}{\letF{x}{t_1}{t_2}}{B}}
    \and
    \inferrule*[right=($\forall$-I)]
    {\hastypel{\Theta,i : \sigma}{\Gamma}{\Delta}{t}{A}}
    {\hastypel{\Theta}{\Gamma}{\Delta}{\Lambda (i:\sigma).t}{\forall (i:\sigma).A}}
    \and
    \inferrule*[right=($\forall$-I)]
    {\hastypei{\Theta}{s}{\sigma} \\
      \hastypel{\Theta}{\Gamma}{\Delta}{t}{\forall (i:\sigma).A}}
    {\hastypel{\Theta}{\Gamma}{\Delta}{t_s}{\sub{s}{i}{A}}}
    \and
    \inferrule*[right=($\exists$-I)]
    {\hastypei{\Theta}{s}{\sigma} \\
      \hastypel{\Theta}{\Gamma}{\Delta}{t}{\sub{s}{i}{A}}}
    {\hastypel{\Theta}{\Gamma}{\Delta}{\exterm{s}{t}}{\exists(i:\sigma).A}}
    \and
    \inferrule*[right=($\exists$-E)]
    {\hastypel{\Theta}{\Gamma}{\Delta_1}{t_1}{\exists (i:\sigma).A} \\
      \hastypel{\Theta,s:\sigma}{\Gamma}{\Delta_2,a : \sub{s}{i}{A}}{t_2}{B}}
    {\hastypel{\Theta}{\Gamma}{\Delta_1,\Delta_2}{\letexists{s}{a}{t_1}{t_2}}{B}}
    \and
    \inferrule*[right=(select)]
    {\hastypel{\Theta}{\Gamma}{\Delta_1}{t_1}{\Event A} \\
      \hastypel{\Theta}{\Gamma}{\Delta_2}{t_2}{\Event B} \\
      \hastypel{\Theta}{\Gamma}{a:A,t_2 : \Event B}{t_1'}{\Event C}\\
      \hastypel{\Theta}{\Gamma}{b:B,t_1 : \Event A}{t_2'}{\Event C}}
    {\hastypel{\Theta}{\Gamma}{\Delta_1,\Delta_2}{\selectterm{t_1}{a}{t_1'}{t_2}{b}{t_2'}}{\Event C}}
    \and
    \inferrule*[right=(delay)]
    {\hastypei{\Theta}{\tau}{\Time} \\
      \Delta' = \attime{\Delta}{\tau} \\
      \hastypel{\Theta}{\Gamma}{\Delta'}{t}{A}}
    {\hastypelt{\Theta}{\Gamma}{\Delta}{t}{\tau}{A}}
    \and
    \inferrule*[right=($\LUnit_\tau$-E)]
    {\hastypei{\Theta}{\tau}{\Time} \\
      \hastypelt{\Theta}{\Gamma}{\Delta_1}{t_1}{\tau}{\LUnit} \\
      \hastypel{\Theta}{\Gamma}{\Delta_2}{t_2}{B}}
    {\hastypel{\Theta}{\Gamma}{\Delta_1,\Delta_2}{\letunitt{\tau}{t_1}{t_2}}{B}}
    \and
    \inferrule*[right=($\otimes_\tau$-E)]
    {\hastypei{\Theta}{\tau}{\Time} \\
      \hastypelt{\Theta}{\Gamma}{\Delta_1}{t_1}{\tau}{A\otimes B} \\
      \hastypel{\Theta}{\Gamma}{\Delta_2,a:_\tau A,b:_\tau B}{t_2}{C}}
    {\hastypel{\Theta}{\Gamma}{\Delta_1,\Delta_2}{\letpairt{a}{b}{\tau}{t_1}{t_2}}{C}}
  \end{mathpar}
    \caption{Linear Typing rules}
  \label{fig:linear_typing_rules}
\end{figure}
The first non-standard rule is for $\Diamond$. The introduction and
elimination rules follow from the fact that $\Diamond$ is a non-strong
monad. More interesting is the $\sym{select}$ rule. Here we see the
formal rule corresponding to the informal explanation in
\autoref{sec:language}. The important thing here is that we can not
use any previously defined linear variable when typing $t_1'$ and
$t_2'$, since we do not actually know \emph{when} the typing
happens. Note, we can see the $\sym{select}$ rule as a binary version
of the $\Diamond$ let-binding. This could additionally be extended to
a $n$-ary version, but we do not do this in our core calculus. The
rules for $A \at \tau$ shows how to move between the judgment
$t : A \at \tau$ and $t :_\tau A$. That is, moving from knowing in the
current timestep that $t$ will have the type $A$ at time $\tau$ and
knowing at time $\tau$ that $t$ has type $A$. The ($\F$-I), ($\F$-E),
($\G$-I) and ($\G$-E) rules show the adjoint structure of the
language. The ($\G$-I) rule takes a closed linear term of type $A$ and
gives it the Cartesian type $\G\;A$. Note, because it has no free
linear variables, it is safe to duplicate. The ($\G$-E) rule lets us
get an $A$ without needing any linear resources. Conversely, the
($\F$-I) rule embeds a intuitionistic term into the linear fragment
and the ($\F$-E) rule binds an intuitionistic variable to let us
freely use the value. The quantification rules ($\exists$ and
$\forall$) should also be familiar, except for the additional
contexts.  The (\sym{Delay}) rule shows what happens when we actually
\emph{know} the timestep. The important part is
$\Delta' = \Delta\downarrow^\tau$ which means two things. One, all the
variables in $\Delta$ are on the form $a :_\tau A$, i.e., judgments at
time $\tau$ and two, we shift $\Delta$ into the future such that all
the variables of $\Delta'$ is of the form $a : A$. The way to
understand this is, if all the variables in $\Delta$ are typed at time
$\tau$ and the conclusion is at time $\tau$, it is enough to ``move
to'' time $\tau$ and then type w.r.t that timestep. Finally, we have
($\LUnit_\tau$-E) and ($\otimes_\tau$-E). These allow us to work with
linear unit and products at time $\tau$. These are added explicitly
since they can not be derived by the other rules, and are needed for
typing certain kinds of programs, e.g., see the typing on
$\mathit{turnRedOnClick}$ below.

\subsection{Unfolding Events to Exists}
The type system as given above contains both $\Diamond A$ and
$A \at k$, as two different way to handle time. The former denotes
that something of type $A$ will arrive at \emph{some} point in the
future, whereas the latter denotes that something of type $A$ arrives
at a \emph{specific} point in the future. The strength of $\Diamond$
is that is gives easy and concise typing rules, whereas the strength
of $A \at k$ is that it allows for a more precise usage of time. To
connect these two, we add the linear isomorphism
$\Diamond A \cong \exists k. A \at k$ to our language, which is
witnessed by $\sym{out}$ and $\sym{into}$, as part of the widget
API. This isomorphism is true semantically, but can not be derived in
the type system. In particular, this isomorphism allows the
$\sym{select}$ rule to be given with $\Diamond$, while still allowing
the use timesteps when working with the resulting event. If we were
to give the equivalent definition using timesteps, one would need to
have some sort of \emph{constraint system} for deciding which events
happens first. Avoiding such a constraint systems also allows for a
much simpler implementation, as everything is our type system can be
inferred.

\subsection{Meta-theory of Substitution}
The meta-theory of $\langname$ is given in the form of a series of
substitution lemmas. Since we have three different contexts, we will
end up with six different substitutions into terms. The Cartesian to
Cartesian, Cartesian to linear and linear to linear are the usual
notion of mutual recursive substitution. More interesting is the
substitution of indices into Cartesian and linear terms and types. We
prove the following lemma, showing that typing is preserved under
index substitution:
\begin{lemma}[Preservation of Typing under Index Substitution]
  \begin{mathpar}
    \inferrule*
    {\zeta : \Theta' \to \Theta \\
      \hastypec{\Theta}{\Gamma}{e}{X}}
    {\hastypec{\Theta'}{\zeta(\Gamma)}{\zeta(e)}{\zeta(X)}}
    \and
    \inferrule*
    {\zeta : \Theta' \to \Theta \\
      \hastypelt{\Theta}{\Gamma}{\Delta}{t}{\tau}{A}}
    {\hastypelt{\Theta'}{\zeta(\Gamma)}{\zeta(\Delta)}{\zeta(t)}{\tau}{\zeta(A)}}
  \end{mathpar}
\end{lemma}
Both are these (and all other cases for substitution) are proved by a
lengthy but standard induction over the typing tree. See the technical
appendix for full proofs of all six substitution lemmas.

\subsection{Typing Example}
In the following, we go through the formal typing of the
$\mathit{turnRedOnClick}$ example from \autoref{sec:language}. In the
below, we have annotated each line with the contents of the index
context (omitting the $i : \Id$ that is given upfront) and the linear
context.

\turnRedOnClickTyping
\noindent
Most of the above is simple application of elimination rules. In line
4, we add the indices variable $x : \Time$ to the index context. Note
in particular the use of $\LUnit_\tau-E$ in line 6 to discharge
$c_2 :_x \LUnit$.  The point where we actually modify the widget is in
line 9 and 10, where we have the following typing:
\begin{mathpar}
  \inferrule*
  {\inferrule*
    {\Delta_1 \vdash w_3 :_x \sym{Widget}\;i}
    {\Delta_1 \vdash (\sym{setColor}\;(\sym{F}\;\sym{Red})\; w_3) \at x : \sym{Widget}\;i \at x}
    \\
    \inferrule*
    {\Delta_2 \vdash p : \sym{Prefix}\;i\;x \\ \Delta_3 \vdash w_4 : \sym{Widget}\;i \at x}
    {\Delta_2,\Delta_3 \vdash \sym{join}\;i\;x \; (p,w_4) : \sym{Widget}\;i}}
  {\Delta_1,\Delta_2 \vdash \mathbf{let}\; w_4 \; = \; (\sym{setColor}\;(\sym{F}\;\sym{Red})\; w_3) \at x\;\mathbf{in}\;\sym{join}\;i\;x\;(p,w_4): \sym{Widget}\;i}
\end{mathpar}
where
$\Delta_1 = w_3 :_x \sym{Widget}\;i, \Delta_2 = p :
\sym{Prefix}\;i\;x$ and $\Delta_3 = w_4 : \sym{Widget}\;i \at x$.

\subsection{Logical Interpretation}

Our language has a straightforward logical interpretation.

The logic corresponding to the Cartesian fragment is a propositional
intuitionistic logic, following the usual Curry--Howard interpretation.
The logic corresponding to the substructural part of the language is a
linear, linear temporal logic. The single-use condition on variables
means that the syntax and typing rules correspond to the rules of
intuitionistic linear logic (i.e., the first occurrence of linear in
``linear, linear temporal''). However, we do not have a comonadic
exponential modality $!A$ as a primitive. Instead, we follow the
Benton--Wadler approach~\cite{benton95mixed,Benton1996LinearLM} and
decompose the exponential into the composition of a pair of adjoint
functors mediating between the Cartesian and linear logic.

In addition to the Benton--Wadler rules, we have a temporal modality
$\Diamond{A}$, which corresponds to the eventually modality of linear
temporal logic (i.e., the second occurrence of ``linear'' in ``linear,
linear temporal logic''). This connective is usually written $F\,A$ in
temporal logic, but that collides with the $\sym{F}$ modality of the
Benton--Wadler calculus. Therefore we write it as $\Diamond{A}$ to
reflect its nature as a possibility modality (or monad). In our
calculus, the axioms of S4.3 are derivable:
\begin{align*}
  (T) &: A \multimap \Diamond A \\
  (4) &:\Diamond\Diamond A \multimap \Diamond A \\
  (.3)&: \Diamond(A \otimes B) \multimap \Diamond((\Diamond A \otimes B) \oplus \Diamond(A \otimes \Diamond B) \oplus \Diamond(A \otimes B))
\end{align*}

Note that because the ambient logic is linear, intuitionistic
implication $X \to Y$ is replaced with the linear implication
$A \multimap B$, and intuitionistic conjunction $X \wedge Y$ is
replaced with the linear tensor product $A \otimes B$. It is easy to
see that the first two axiom corresponds to the monadic structure of
$\Diamond$, and the .3 axiom corresponds to the $\sym{select}$ rule
(with our syntax for $\sym{select}$ corresponding to immediately
waiting for and then pattern-matching on the sum type).  In the
literature, the .3 axiom is often written in terms of the box modality
$\Box{A}$~\cite{blackburn2002modal}, but we present it here in a
(classically) equivalent formulation mentioning the eventually
modality $\Diamond{A}$.

We do not need to offer an additional explicit box modality $\Box A$,
since the decomposition of the exponential $\sym{F}(\sym{G} A)$ from
the linear-non-linear calculus serves that role.

In our system, \emph{we do not want to offer} the next-step operator
$\Later A$. Since we want to model asynchronous programming, we do not
want to include a facility for permitting programmers to write
programs which wake up in a specified amount of time. Instead, we only
offer an iterated version of this connective, $A \at n$, which can be
interpreted as $\Later^n A$, and our term syntax does not have any
numeric constants which can be used to demand a specific delay.

Finally, the universal and existential quantifiers (in both the
intuitionistic and linear fragments) are the usual quantifier
rules for first-order logic.

%%% Local Variables:
%%% mode: latex
%%% TeX-master: t
%%% End:

%% file: article-tex/formal_semantics.tex
In this section we will present a denotational model for $\langname$.
The model is a linear-non-linear (LNL)
hyperdoctrine~\cite{maietti00iltcats,haim16linhyper} with the
non-linear part being $\Set$ and the linear part being the category of
internal relations over a suitable ``reactive'' category. The
hyperdoctrine structure itself is used to interpret the quantification
over indices.  It many ways this model is entirely standard, and the
most interesting thing is the reactive base category and the
interpretation of widgets. It is well known that any symmetric
monoidal closed category (SMCC) models multiplicative intuitionistic
linear logic (MILL), and it is similarly well known that the category
of relations over $\Set$ can be give the structure of a SMCC by using
the Cartesian product as both the monoidal product and monoidal
exponential. This construction lift directly to any category of
internal relations over a category that is suitably ``\Set-like'',
i.e., a topos. Our base category is a simple presheaf category, and
hence, we use this construction to model the linear fragment of
$\langname$.

\subsection{The Base Reactive Category}
The base reactive category is where the notion of time will arise and
is it this notion that will be lifted all the way up to the LNL
hyperdoctrine. The simplest model of ``time'' is $\Set^\nats$, which
can be understood as ``sets through
time''~\cite{maclane1994sheaves}. This can indeed by used as a model
for a reactive setting, but for our purposes it is too simple, and
further, depending on which ordering is considered for $\nats$, may
have undesirable properties for the reactive setting. Instead, we use
the only slightly more complicated $\Set^{\nats+1}$, henceforth
denoted $\rcat$, where the ordering on $\nats + 1$ is the discrete
ordering on $\nats$ and $1$ is related to everything else. Adding this
``point at infinity'' allows global reasoning about objects, an
intuition that is further supported by the definition of the
sub-object classifier below. Further, this model is known to be able
to differentiate between least and greatest
fixpoints~\cite{graulund2018}, and even though we do not use this for
$\langname$, we consider it a useful property for further work (see
\autoref{sec:related_work}). Objects in $\rcat$ can be visualized as
\begin{center}
  $A =$
  \begin{tikzcd}
    & A_\infty \arrow[dl,"\pi_1",swap]\arrow[d,"\pi_2"]\arrow[dr] \\
    A_0 & A_1 & \cdots
  \end{tikzcd}
\end{center}
We can think of $A_\infty$ as the global view of the object and $A_n$
as the local view of the object at each timestep. Morphisms are
natural transformations between such diagrams and the naturality
condition means that having a map from $A_\infty$ to $B_\infty$ must
also come with coherent maps at each timestep.

In $\rcat$ we define two endofunctors, which can be seen as describing
the passage of time:
\begin{definition}
  \label{def:rcat_later_prev}
  We define the \emph{later} and \emph{previous} endofunctors on
  $\rcat$, denoted $\Later$ and $\Prev$, respectively:
  \begin{align*}
    (\Later A)_n :=
    \begin{cases}
      1 & n = 0 \\
      A_{n'} & n = n'+1 \\
      A_\infty & n = \infty
    \end{cases}
             && (\Prev A)_n :=
                \begin{cases}
                  A_{n+1} & n \neq \infty \\
                  A_\infty & n = \infty
                \end{cases}
  \end{align*}
\end{definition}
Note that when we apply the later functor, the global view does not
change, but the local views are shifted forward in time.

\begin{theorem}
  The later and previous endofunctors form an adjunction:
  \begin{equation*}
    \Prev \vdash \Later
  \end{equation*}
\end{theorem}
\begin{proof}
  The proof follows easily from an examination of the appropriate diagrams.
\end{proof}

\begin{definition}
  \label{def:rcat_subobject}
  The sub-object classifier, denoted $\Omega$, in $\rcat$ is the object
  \begin{center}
    \begin{tikzcd}
      & \mathcal{P}(\nats) + 1 \arrow[dl]\arrow[d]\arrow[dr]\\
      \set{0,1} & \set{0,1} & \cdots
    \end{tikzcd}
  \end{center}
\end{definition}
For each $n \in \nats$, $\Omega_n$ denotes whether a given proposition
is true at the $n$th timestep. $\Omega_\infty$ gives the ``global
truth'' of a given proposition. The left injection is some subset of
$\nats$ that denotes at which points in time something is true. The
right injection denotes that something is true ``at the limit'', and
in particular, also at all timesteps. Note, a proposition can be true
at all timesteps but not at the limit. This extra point at infinity is
precisely what allows us the differentiate between least and greatest
fixpoints.

\subsection{The Category of Internal Relations}
To interpret the linear fragment of the language, we will use the
category of internal relations on $\rcat$. Given two objects $A$ and
$B$ in $\rcat$, an \emph{internal relation} is a sub-object of the
product $A \times B$. This can equivalently by understood as a map
$A \times B \to \Omega$. The category of internal relations in the
category where the objects are the objects of $\rcat$ and the
morphisms $A \to B$ are internal relations $A \times B \to \Omega$ in
$\rcat$. We denote the category of internal relations as $\relcat$.

\begin{definition}
  We define a monoidal product and monoidal exponential on $\relcat$ as
  \begin{align*}
    A \otimes B = A \times B && A \multimap B = A \times B
  \end{align*}
\end{definition}

\begin{theorem}
  Using the above definition of monoidal product and exponential,
  $\relcat$ is a symmetric monoidal closed category.
\end{theorem}
\begin{proof}
  All of the properties of the monoidal product and exponential
  follows easily. Consider the evaluation map
  $(A\multimap B) \otimes A \to_{\relcat} B$. By definition this is a relation
  $(A \times B) \times A \sim_\rcat B$, which is a map
  $((A \times B) \times A) \times B \to_\rcat \Omega$. We define this map to
  be ``true'' for tuples $(((a,b),a'),b')$ with $a =_\rcat a' \land b =_\rcat b'$.
\end{proof}

\begin{theorem}
  There is an adjunction in $\relcat$:
  \begin{equation*}
    \Prev \vdash \Later
  \end{equation*}
  where $\Prev$ and $\Later$ are the lifting of the previous and later
  functors from $\rcat$ to $\relcat$.
\end{theorem}

\begin{definition}
  We define the \emph{iterated later modality} or the ``at'' connective
  as a successive application of the later modality.
  \begin{align*}
    \Later^0 A &= A \\
    \Later^{(k+1)} A &= \Later(\Later^k A)
  \end{align*}
  and we will alternatively write $A \at k$ to mean $\Later^k A$.
\end{definition}

\begin{definition}
  We define the \emph{event} functor on $\relcat$ as an application of
  the iterated later modality.
  \begin{align*}
    \Diamond A &: \relcat \to \relcat \\
    (\Diamond A)_\infty &= A_\infty \\
    (\Diamond A)_n &= \Sigma(k:\nats).(\Later^k\;A)_n
  \end{align*}
\end{definition}
The event functor additionally carries a monadic structure (see~\cite{szamozvancev2018} and the
technical appendix).

\begin{theorem}
  We have the following isomorphism for any $A$
  \begin{align*}
    \Diamond A \cong \Sigma(n:\nats).A \at n
  \end{align*}
\end{theorem}
\begin{proof}
  This follows immediately by the two preceding definitions.
\end{proof}

\begin{theorem}
  We have the following adjunctions between $\Set$, $\rcat$ and
  $\relcat$:
  \begin{center}
    \begin{tikzcd}
      \Set \arrow[rr, "\Delta", bend left] & \perp &
      \rcat  \arrow[rr, "I", bend left] \arrow[ll, "\sym{lim}", bend left] & \perp &
      \relcat \arrow[ll, "P", bend left]
    \end{tikzcd}
  \end{center}
  where $\Delta$ is the constant functor, $\sym{lim}$ is the limit
  functor, $I$ is the inclusion functor and $P$ is the image functor.
\end{theorem}

\begin{corollary}
  \label{cor:set_relcat_adjunction}
  The above adjunction induces an adjunction between $\Set$ and $\relcat$.
\end{corollary}

\subsection{The Widget Object}
One of the most important objects in $\relcat$ is the \emph{widget}
object. This object will be used to interpret widgets and
prefixes. The widget object will be defined with respect to an ambient
notion of identifiers, which we will denote $\Id$. These will be part
of the hyperdoctrine structure define below, and for now, we will just
assume such an object to exists. We will also use a notion of
timesteps internal to the widget object. Note that this timestep is
different from the abstract timestep used for defining $\relcat$, but
are related as defined below. We denote the abstract timesteps with
$\Time$.

Before we can define the widget object itself, we need to define an
appropriate object of commands. In our minimal Widget API, the only
\emph{semantic} commands will be $\sym{setColor}, \sym{onClick}$ and
$\sym{onKeypress}$. The rest of the API will be defined as morphisms
on the widget object itself. To work with the semantics commands, we
additionally need a \emph{compatibility} relation. This relation
describes what commands can be applied at the same time. In our
setting this relation is minimal, but can in principle be used to
encode whatever restrictions is needed for a given API.

\begin{definition}
  We define the command object as
  \begin{align*}
    \Cmd = \set{(\sym{setColor},color), \sym{onClick}, \sym{onKeypress}}
  \end{align*}
  where $color$ is some element of an ``color'' object.  We define the
  compatibility relations as
  \begin{align*}
    cmd \comp cmd' \text{ iff } cmd = (\sym{setColor},c) \Rightarrow cmd' \neq (\sym{setColor},c')
  \end{align*}
\end{definition}
The only non-compatible combination of commands is two application of
the $\sym{setColor}$ command, the idea being that you can not set the
color twice in the same timestep.

We can now define the widget and prefix objects
\begin{definition}
  The widget object, denoted $\Widget$, is indexed by $i \in \Id$ and
  is defined as
  \begin{align*}
    \Widget_\infty \;i &= \setcom{(w,i)}{ w \in \mathcal{P}(\Time \times \Cmd), (t,c) \in w \land (t,c') \in w \to c \comp c'}\\
    \Widget_n \;i  &=  \setcom{(w,i) \subset \Widget_\infty\; i}{\forall(t,c) \in w, t \leq n}
  \end{align*}
  The prefix object, denoted $\sym{Prefix}$, is indexed by $i \in \Id$ and
  $t \in \Time$ and is defined as
  \begin{align*}
    \sym{Prefix}_\infty \;i \;t &= \setcom{(P,i) \subset \Widget_\infty \;i}{\forall (t',c) \in P, t' \leq t} \\
    \sym{Prefix}_n \;i \;t &=
                             \begin{cases}
                               \setcom{(P,i) \subset \sym{Prefix}_\infty \;i \;t}{\forall (t',c) \in P, t' \leq n} & n < t \\
                               \LUnit & \text{otherwise}
                             \end{cases}
  \end{align*}
\end{definition}
The widget object is basically a collection of times and commands
keeping track of what has happened to the widget at various times. One
can think of an \emph{logbook} with entries for each time step. At the
point at infinity, the ``global'' behavior of the widget is defined,
i.e., the full logbook of the widget. For each $n$, $\Widget_n$ is
simply what has happened to the widget so far, i.e., a truncated
logbook. The prefix object is a widget object that is only defined up
to some timestep, and is the unit after that.

Observe there is a semantic difference between the widget where the
color is set only once, and the widget where the color is set at every
timestep, and this reflects a real difference in actual widget
behavior. The difference between $\mathit{turnRedOnClick}\;w$ and
$\mathit{keepTurningRed}\;w$ is that if the former is later set to be
blue, it will remain blue, whereas the latter will turn back to being
red.

To work with widgets we define two ``restriction'' maps, which are
used later for the interpretations.
\begin{definition}
  We define the following on widgets and prefixes
  \begin{align*}
    &\sym{shift} \; t : \sym{Widget}\;i \to_{\relcat} \sym{Widget}\;i
    && \sym{prefix} \; t \; i : \sym{Widget}\;i \to_{\relcat} \sym{Prefix}\;i\;t \\
    &(\sym{shift} \;t \; W)_n = \setcom{(t' - t,c)}{(t',c) \in W \land t \leq t'}
    &&
    (\sym{prefix} \; t \; i \; W)_n =
    \begin{cases}
      \setcom{(t',c) \in W}{t' < t} & n < t \\
      \LUnit & n \geq t
    \end{cases}
  \end{align*}
\end{definition}
The intuition behind these is that $\sym{prefix} \; t \; i$ ``cuts off'' the
widget after $t$, giving a prefix, whereas $\sym{shift} \; t$ shifts
forward all entries in the widget by $t$.

Using the above, we can now define the $\sym{split}$ and $\sym{join}$
morphisms. These are again given w.r.t ambient $\Id$ and $\Time$
objects, which will be part of the full hyperdoctrine structure:
\begin{definition}
  We define the following morphisms on the widget object
  \begin{align*}
    &\sym{split}\;i\;t : \Widget\;i \to_{\relcat} \sym{Prefix}\;i\;t \otimes \Widget\;i \at t
    && \sym{join}\;i\;t : \sym{Prefix}\;i\;t \otimes \Widget\;i \at t \to_{\relcat} \Widget\;i \\
    &(\sym{split}\;i\;t\;w)_n = (\sym{prefix}\;t\;i\;w,\sym{shift}\;t\;w)_n
    &&
    (\sym{join}\;i\;t\;(p,w))_n =
                              \begin{cases}
                                p_n & n < t \\
                                w_{n-t} & n \geq t
                              \end{cases}
  \end{align*}
\end{definition}

\subsection{Linear-non-linear Hyperdoctrine}
So far we have not explained in details how to model the quantifiers
in our system. To do this, we use the notion of a
\emph{hyperdoctrine}~\cite{lawvere1969adj}. For ordinary first-order
logic, this is a functor from a category of contexts and substitutions
to the category of Cartesian closed categories, with the idea being that we have
one CCC for each valuation of the free first-order variables.

As our category of contexts, we use a Cartesian category that can
interpret our index objects, namely $\Time$ and $\Id$, where the
former is interpreted as $\nats + 1$ and the latter as $\nats$. In our
case, both $\Set$ and $\relcat$ are themselves hyperdoctrines w.r.t to
this category of contexts, the former a first-order hyperdoctrine and
the latter a multiplicative intuitionistic linear logic (MILL)
hyperdoctrine. Together these form a linear-non-linear hyperdoctrine
through the adjunction given in
\autoref{cor:set_relcat_adjunction}. Formally, we have

\begin{definition}
  A linear-non-linear hyperdoctrine is a MILL hyperdoctrine $L$
  together with a first-order hyperdoctrine $C$ and a fiber-wise
  monoidal adjunction $F : L \leftrightarrows C : G$.
\end{definition}

\begin{theorem}
  The categories $\Set$ and $\relcat$ form a linear-non-linear
  hyperdoctrine w.r.t the interpretation of the indices objects, with
  the adjunction given as in \autoref{cor:set_relcat_adjunction}.
\end{theorem}
We refer the reader to the accompanying technical appendix for the
full details.

\subsection{Denotational Semantics}
We the above, we have enough structure to give an interpretation of
$\langname$. Again, most of this interpretation is standard in the use
of the hyperdoctrine structure, and we interpret $\Diamond$ in the
obvious way using the linear hyperdoctrine structure on $\relcat$. As
an example, we sketch the interpretation of the widget object and the
$\sym{setColor}$ command below.

\begin{definition}
  We interpret the $\sym{Widget}\;i$ and $\sym{Prefix}\;i$ types using
  the widget and prefix objects:
  \begin{align*}
    \llbracket \Theta \vdash \sym{Widget}\;i \rrbracket
    &= \sym{Widget} \; \llbracket \Theta \vdash_s i : \sym{Id} \rrbracket \\
    \llbracket \Theta \vdash \sym{Prefix}\;i\;t \rrbracket
    &= \sym{Prefix} \; \llbracket \Theta \vdash_s i :\sym{Id} \rrbracket \; \llbracket \Theta \vdash_s t : \sym{Time} \rrbracket
  \end{align*}
  and we interpret the $\sym{setColor}$ commands as:
  \begin{align*}
    \llbracket \sym{setColor} : \forall (i:\sym{Id}), \sym{Widget}\;i \otimes \sym{F}\;\sym{Color} \multimap \sym{Widget}\;i \rrbracket = \\
    \setcom{w \cup_W \set{(0,(\sym{setColor},col))}}{w \in \llbracket \sym{Widget}\;i \rrbracket, col \in \llbracket \sym{Color} \rrbracket}
  \end{align*}
  where $\cup_W$ is a ``widget union'', which is a union of sets such
  that identifiers indices and compatibility of commands are respected
\end{definition}
This interpretation shows that a widget is indeed a logbook of
events. Using the \sym{setColor} command simply adds an entry to the
logbook of the widget. Note we only set the color in the current
timestep. To set the color in the future, we combine the above with
appropriate uses of splits and joins.
The interpretation of
$\sym{split}$ and $\sym{join}$ are done using their semantic
counterparts, and the interpretation of $\sym{onClick}$ and
$\sym{onKeypress}$ are done, using our non-deterministic semantics, by
associating a widget with \emph{all possible occurrences} of the
corresponds event.

\subsection{Soundness of Substitution}
Finally, we prove that semantic substitution is sound w.r.t syntactic
substitution. As with the proofs of type preservation for syntactic
substitution, there are several cases for the different kinds of
substitution, but the main results is again concerned with
substitution of indices:
\begin{theorem}
  Given $\zeta : \Theta' \to \Theta$, $\hastypec{\Theta}{\Gamma}{e}{X}$ and $\hastypel{\Theta}{\Gamma}{\Delta}{t}{A}$ then
  \begin{align*}
    \llbracket \zeta \rrbracket \; \llbracket \hastypec{\Theta}{\Gamma}{e}{X} \rrbracket
    &= \llbracket \hastypec{\Theta'}{\zeta(\Gamma)}{\zeta(e)}{\zeta(X)} \rrbracket
    \\
    \llbracket \zeta \rrbracket \; \llbracket \hastypel{\Theta}{\Gamma}{\Delta}{t}{A} \rrbracket
    &= \llbracket \hastypel{\Theta'}{\zeta(\Gamma)}{\zeta(\Delta)}{\zeta(t)}{\zeta(A)}\rrbracket
  \end{align*}
\end{theorem}
Proofs for all six substitutions lemmas can be found in the technical
appendix.

%%% Local Variables:
%%% mode: latex
%%% TeX-master: t
%%% End:

%% file: article-tex/related_work.tex
%% Logical correspondence

Much work has sought to offer a logical perspective on FRP via the
Curry--Howard
correspondence~\cite{krishnaswami2011ultrametric,Jeltsch2012,jeffrey2012,jeltsch2013temporal,krishnaswami13frp,cave14fair,bahr2019simply}. As
mentioned earlier, most of this work has focused on calculi that have
a Nakano-style later modality~\cite{nakano2000}, but this has the
consequence that it makes it easy to write programs which wake up on
every clock tick.

%% Related implementation stuff

In this paper, we remove the explicit next-step modality from the
calculus, which opens the door to a more efficient implementation
style based on the so-called ``push'' (or event/notification-based)
implementation style. Elliott~\cite{Elliott2009-push-pull-frp} also
looked at implementing a push-based model, but viewed it as an
optimization rather than a first-class feature in its own right. We
also hope to use an effect system to track when reflows and redraws
occur, which should make it easier to keep track of when potentially
expensive UI operations are taking place. Moreover, we can extend the
widget semantics and compatibility relation to track these events,
which should let us test if we can easily put the domain knowledge of
browser developers into our semantic model.

In future work, we plan on implementing a language based upon this
calculus, with the idea that we can compile to Javascript, and
represent widgets with DOM nodes, and represent the $\Diamond{A}$ and
$A \at n$ temporal connectives using doubly-negated callback types (in
Haskell notation, \texttt{Event A = (A -> IO ()) -> IO ()}).  This
should let us write GUI programs in a convenient functional style,
while generating code which attaches callbacks and does mutable
updates in the same style that a handwritten GUI program would use.

%% Models and semantics

Our model, in terms of $\Set^{\nats+1}$, can be seen as a model of LTL
that has been enriched from being about time-indexed true-values to
time-indexed sets. The addition of the global view or point at
infinity has the benefit that it enables us to give a model that
distinguishes between least and greatest fixed
points~\cite{graulund2018} (i.e., inductive and coinductive types),
unlike in models of guarded recursion where guarded types are
bilimit-compact~\cite{Birkedal11firststeps}.  This will let us use the
idea that temporal liveness and safety properties can be encoded in
types using inductive and coinductive
types~\cite{cave14fair,bahr2020diamonds}.

%% Future calculi

One interesting recent development in natural deduction systems for
comonadic modalities is the introduction of the so-called
``Fitch-style'' calculi~\cite{BirkedalL:gdtt-conf,clouston2018fitch},
which offer an alternative to the Pfenning--Davies pattern-style
elimination~\cite{pfenning2001} for the box comonad.  These
calculi have seen successful application to reactive programming
languages~\cite{bahr2019simply}, and one interesting question is
whether they extend to Benton--Wadler adjoint calculi as well -- i.e.,
can the $\F(X)$ modality be equipped with a direct style eliminator?

%%% Local Variables:
%%% mode: latex
%%% TeX-master: t
%%% End: